\documentclass[11pt]{article}

\usepackage{amsmath}
\usepackage{amsthm}
\usepackage{graphicx} 
\usepackage{array} 
\usepackage{bbm}
\usepackage{thm-restate}

\usepackage{fancybox}

\newenvironment{fminipage}%
  {\begin{Sbox}\begin{minipage}}%
  {\end{minipage}\end{Sbox}\fbox{\TheSbox}}

\DeclareMathOperator*{\poly}{poly}
\DeclareMathOperator{\diag}{diag}

\newcommand{\norm}[1]{\left\lVert#1\right\rVert}

\usepackage{amsmath, amssymb, amsfonts, verbatim}
\usepackage{hyphenat,epsfig,subfigure,multirow}
\usepackage{multicol}
\usepackage{nicefrac}
\usepackage{paralist}
\usepackage{enumitem}
\usepackage{tikz}
\usetikzlibrary{positioning}

\usepackage[ruled]{algorithm2e}

\DeclareFontFamily{U}{mathx}{\hyphenchar\font45}
\DeclareFontShape{U}{mathx}{m}{n}{
      <5> <6> <7> <8> <9> <10>
      <10.95> <12> <14.4> <17.28> <20.74> <24.88>
      mathx10
      }{}
\DeclareSymbolFont{mathx}{U}{mathx}{m}{n}
\DeclareMathSymbol{\bigtimes}{1}{mathx}{"91}

\usepackage{tcolorbox}
\tcbuselibrary{skins,breakable}
\tcbset{enhanced jigsaw}

\usepackage[normalem]{ulem}
\usepackage[compact]{titlesec}

\definecolor{DarkRed}{rgb}{0.5,0.1,0.1}
\definecolor{DarkBlue}{rgb}{0.1,0.1,0.5}

\usepackage{nameref}
\definecolor{ForestGreen}{rgb}{0.1333,0.5451,0.1333}
\definecolor{Red}{rgb}{0.9,0,0}
\usepackage[linktocpage=true,
	pagebackref=true,colorlinks,
	linkcolor=DarkRed,citecolor=ForestGreen,
	bookmarks,bookmarksopen,bookmarksnumbered]
	{hyperref}
\usepackage[noabbrev,nameinlink]{cleveref}
\crefname{property}{property}{Property}
\creflabelformat{property}{(#1)#2#3}
\crefname{equation}{eq}{Eq}
\creflabelformat{equation}{(#1)#2#3}

\usepackage{bm}
\usepackage{url}
\usepackage{xspace}
\usepackage[mathscr]{euscript}

\usepackage{tikz}
\usetikzlibrary{arrows}
\usetikzlibrary{arrows.meta}
\usetikzlibrary{shapes}
\usetikzlibrary{backgrounds}
\usetikzlibrary{positioning}
\usetikzlibrary{decorations.markings}
\usetikzlibrary{patterns}
\usetikzlibrary{calc}
\usetikzlibrary{fit}
\usetikzlibrary{snakes}

\usepackage{mdframed}

\usepackage[noend]{algpseudocode}
\makeatletter
\def\BState{\State\hskip-\ALG@thistlm}
\makeatother

\usepackage{cite}
\usepackage{enumitem}

\usepackage[margin=1in]{geometry}

\newtheorem{theorem}{Theorem}[section]

\newtheorem{lemma}{Lemma}[section]

\newtheorem{corollary}[lemma]{Corollary}

\newtheorem{definition}[lemma]{Definition}

\newtheorem*{theorem*}{Theorem}
\newtheorem*{claim*}{Claim}
\newtheorem*{proposition*}{Proposition}
\newtheorem*{lemma*}{Lemma}
\newtheorem*{problem*}{Problem}

\crefname{lemma}{Lemma}{Lemmas}
\crefname{claim}{Claim}{Claims}

\newtheoremstyle{restate}{}{}{\itshape}{}{\bfseries}{~(restated).}{.5em}
  {\thmnote{#3}}
\theoremstyle{restate}

\allowdisplaybreaks

\title{
Approximating Random Walks in $\Otil(\log n + \log^2 \kappa)$ Space for
$\kappa$-Conditioned Graphs
}
\author{
Junzhao Yang\\
{\small Carnegie Mellon University}\\
{\small\texttt{junzhaoy@andrew.cmu.edu}}
}
\date{}

\newcommand{\Mbegin}{\left[\begin{matrix}}
\newcommand{\Mend}{\end{matrix}\right]}

\newcommand{\callalg}[1]{\hyperref[alg:#1]{\textsc{#1}}}

\newcommand{\defeq}{\triangleq}
\newcommand{\Otil}{\widetilde{O}}
\newcommand{\R}{\mathbb{R}}

\newcommand{\bone}{\mathbf{1}}
\newcommand{\eps}{\varepsilon}
\newcommand{\wrt}{\ \mathsf{wrt}\ }
\newcommand{\Zhat}{\overline{Z}}
\newcommand{\rhat}{\widehat{r}}

\begin{document}

\pagenumbering{roman}

\maketitle

\begin{abstract}
For $\kappa>1$, a directed graph is
$\kappa$-conditioned if it is
$\kappa$-mixing and its stationary distribution is approximated by the uniform
distribution within a factor of $\kappa$.
We present a deterministic algorithm that approximates the stationary
distribution of a $\kappa$-conditioned graph to inverse polynomial relative
error in $O((\log n + \log^2 \kappa) \log \log n)$ space.
In the regime $\kappa = \exp(\Theta(\log^\alpha n))$ for any
$\alpha \in (0, 2/3)$, our result improves the best-known
$O(\log n \sqrt{\log \kappa} / \sqrt{\log \log n})$ space bound for approximating
$\kappa$-step random walks in general directed graphs by [Hoza, RANDOM 2021].

We release this preliminary version due to recent rumors of LLM-based
progress on related problems and uncertainty about when those results may
appear.
Further implementation details will be provided in a subsequent version.
\end{abstract}

\pagenumbering{arabic}

\section{Introduction}\label{sec:intro}

Approximating random walks is a basic problem in the derandomization of
space-bounded computation, closely related to the question of whether
$\mathsf{RL}=\mathsf{L}$.
For undirected graphs, Reingold~\cite{R08:journal} showed that connectivity can
be decided in deterministic logarithmic space.
For Eulerian graphs, Ahmadinejad et al.~\cite{AKMPSV20} gave a deterministic
$O(\log n\log\log n)$-space algorithm for polynomial-length walks with
inverse-polynomial additive error.
For general directed graphs, the algorithm of Saks and Zhou~\cite{SZ99:journal}
estimates the probabilities of polynomial-length walks to inverse-polynomial
additive error in $O(\log^{3/2}n)$ space.

We study stationary distributions and walk probabilities for
$\kappa$-conditioned directed graphs as an intermediate problem towards proving
$\mathsf{RL}=\mathsf{L}$.

\begin{definition}[$\kappa$-conditioned graph] \label{def:kappa-conditioned}
A column-stochastic random walk matrix $W \in \R^{n \times n}$ is
$\kappa$-conditioned if it satisfies the following conditions:
\begin{itemize}
\item The stationary distribution $s(W)$ satisfies
  $1/(\kappa n)\le s(W)_i\le\kappa/n$ for every
  $i$.
\item $|(W^\kappa-s(W)\bone^\top)_{ij}|\le n^{-10}$ for all
  $i,j$.
\end{itemize}
\end{definition}

Our main result approximates the stationary distribution to $1/\poly(n)$
relative error.

\begin{theorem}\label{thm:main}
For every fixed constant $c>0$, there is a deterministic algorithm that, given
$n$, an integer
$1 < \kappa\le n$, and adjacency-list access to an
$n$-vertex constant-degree out-regular directed graph whose walk matrix
$W$ is
$\kappa$-conditioned, outputs
$\widehat{s}_1,\ldots,\widehat{s}_n$ satisfying
\[
\left|\widehat{s}_i-s(W)_i\right|
\le n^{-c}s(W)_i
\qquad\text{for every }i\in V,
\]
using
\[
O\!\left((\log n+\log^2\kappa)\log\log n\right)
\]
bits of work space.
\end{theorem}

In particular, the space bound is $O(\log n\log\log n)$ when
$\log\kappa\le\sqrt{\log n}$.
Combining the stationary approximation with the Eulerian walk-estimation
algorithm of Ahmadinejad et al.~\cite{AKMPSV20} gives the following consequence.

\begin{corollary}\label{cor:intro-walk}
For every fixed constant $c>0$, there is a deterministic algorithm that, given a
graph and parameters satisfying the assumptions of Theorem~\ref{thm:main},
vertices $i,j\in V$, and an integer
$1\le k\le \poly(n)$, outputs
$\widehat{p}$ satisfying
\[
\left|\widehat{p}-(W^k)_{ij}\right|
\le n^{-c},
\]
using $O\!\left((\log n+\log^2\kappa)\log\log n\right)$ bits of work space.
\end{corollary}


The main construction combines derandomized squaring with a \emph{trimming}
process.
Error accumulates in the plain derandomized squaring, and it fails due to the blow-up of the in-degrees or row sums of the walk matrix power.
Our trimming process rebalances the in-degrees and make the derandomized squaring work, but
requires squared logarithmic space due to multiple rounds of trimming.
We address this issue by defining a trimming step that performs trimming partially, and carefully coordinating the squaring and trimming steps
in a synchronized manner, reducing the overall space complexity.
Finally, a Reingold-style implementation uses a shared vertex register and short local
indices and evaluates the recursion in small space.


\subsection{Related work}

\paragraph{General directed walks.}
Saks and Zhou~\cite{SZ99:journal} gave a deterministic simulation of randomized
space-bounded computation, yielding the classical
$O(\log^{3/2}n)$ space bound for estimating polynomial-length walk
probabilities to inverse-polynomial additive error.
Hoza~\cite{H21} subsequently obtained the bound
$O(\log^{3/2}n/\sqrt{\log\log n})$ using weighted pseudorandom generators
within the Saks--Zhou framework.
These algorithms apply to general directed graphs, without assumptions on
mixing or the stationary distribution.

For shorter walks, their dependence on the walk length is also relevant:
the Saks--Zhou bound for
$\kappa \le\poly(n)$ steps and inverse-polynomial additive error is
$O(\log n\sqrt{\log \kappa})$.
For example, when
$\kappa=\exp(\Theta(\log^\alpha n))$ for a fixed
$\alpha\in(0,2/3)$, this bound for
$\kappa$-step walks is
$O(\log^{1+\alpha/2}n)$, whereas the bound in
Theorem~\ref{thm:main} is
\[
O\!\left((\log n+\log^{2\alpha}n)\log\log n\right).
\]
Our theorem improves the bound of the general setting, but uses both a mixing assumption and a bound on the range of
stationary probabilities.

\paragraph{Recycling state randomness.}
Raz and Reingold~\cite{RR99} derandomize branching programs by recycling the
randomness in their states.
Their unconditional result gives statistical error
$2^{-r}$ for width
$w$ and length
$2^r$, using
$O(\log w)$ space, where
$r=\sqrt{\log w}/(\log\log w)^{c_0}$ for a constant
$c_0>0$, given
$2^r$-factor estimates of intermediate state probabilities.

In graph terms, the estimates approximate the probability of reaching each
vertex at each time, including before the walk mixes.
For a uniform starting distribution, the $\kappa$-conditioning assumption
supplies $\kappa^2$-factor estimates using the uniform distribution.
For $\kappa$-step walks with
$\kappa^{O(1)}$-factor estimates, their construction, combined with the
space-efficient extractors of Raz, Reingold, and Vadhan and Hartman and
Raz~\cite{RRV02:journal,HR03:journal}, approximates walk
probabilities to additive error $\eps\ge n^{-O(1)}$ in space
\[
O\!\left(\left[\log n+\log^2\kappa+
\log\kappa\,\log(1/\eps)\right](\log\log n)^C\right)
\]
for an absolute constant $C>0$.
Our result achieves $O((\log n+\log^2\kappa)\log\log n)$
space for inverse-polynomial additive error, while requiring the additional
mixing-time bound in the $\kappa$-conditioning assumption.

\paragraph{Known stationary distributions and Eulerian graphs.}
Chung, Reingold, and Vadhan~\cite{CRV11:journal} solve directed connectivity
in logarithmic space given a stationary distribution and polynomial mixing time from
the source.
For Eulerian graphs, whose stationary probabilities are determined by the
degrees, Ahmadinejad et al.~\cite{AKMPSV20} estimate polynomial-length walk
probabilities to inverse-polynomial additive error in
$O(\log n\log\log n)$ space.
Our result computes the stationary distribution under the $\kappa$-conditioned
assumptions.



\section{Preliminaries}\label{sec:prelim}

\subsection{Notations}
Let $V$ be a vertex set of size
$n\ge2$.
All logarithms are to base two unless a base is displayed.
We assume an integer $\kappa\ge2$ with
$\kappa\le\poly(n)$.
The vector $\bone$ denotes the all-ones vector in the relevant dimension, and
$\bone_I$ denotes the indicator of a set
$I$.
For nonnegative vectors, $x\approx_K y$ means coordinatewise comparison within a
factor $K$ when
$K\ge1$, and within a factor
$1+K$ when
$0\le K<1$.

Let $W$ be the nonnegative column-stochastic transition matrix of a
constant-degree out-regular directed graph on
$V$.
Thus $W_{ij}$ is the transition probability from
$j$ to
$i$.
Write $s(W)$ for its positive stationary probability vector and
$S(W)\defeq\diag(s(W))$.
We recall the $\kappa$-conditioning assumptions from
Definition~\ref{def:kappa-conditioned}:
\begin{itemize}
\item $s(W)\approx_\kappa \frac{1}{n}\bone$, equivalently
  $1/(\kappa n)\le s(W)_i\le\kappa/n$ for every
  $i$.
\item $|(W^\kappa-s(W)\bone^\top)_{ij}|\le n^{-10}$ for all
  $i,j$.
\end{itemize}

\subsection{Norms and Approximation}
We use $\norm{x}_2$ and
$\norm{A}_2$ for the Euclidean and spectral norms, respectively.
The vector norm $\norm{x}_1$ is the sum of absolute coordinates and
$\norm{x}_\infty$ is their maximum.
For matrices, $\norm{A}_1$ and
$\norm{A}_\infty$ denote the maximum absolute column sum and row sum,
respectively.
For a positive definite matrix $M$, define the weighted vector norm by
\[
\norm{x}_M\defeq\sqrt{x^\top Mx},
\]
and the induced operator norm by
\[
\norm{A}_M
\defeq\sup_{x\ne0}\frac{\norm{Ax}_M}{\norm{x}_M}
=\norm{M^{1/2}AM^{-1/2}}_2.
\]
The stationary norm is obtained by taking $M=S(W)^{-1}$.
These stationary weights are used only in the analysis and are not supplied to
the algorithm.

We also compare matrices using approximation with respect to error matrices
from~\cite{APPSV23}, together with a one-sided variant.
\begin{definition}\label{def:approximation}
Let $A,B\in\mathbb R^{n\times m}$, let
$E\in\mathbb R^{n\times n}$ and
$F\in\mathbb R^{m\times m}$ be positive semidefinite error matrices, and let
$\eps\ge0$.
We say that $A$ is an
$\eps$-approximation of
$B$ with respect to
$(E,F)$, written
$A\approx_\eps B\wrt(E,F)$, if
\[
|x^\top(A-B)y|\le\eps\sqrt{x^\top Ex}\sqrt{y^\top Fy}
\quad\text{for all }x\in\mathbb R^n,\ y\in\mathbb R^m.
\]
We say that $A$ is a one-sided
$\eps$-approximation of
$B$ with respect to
$(E,F)$, written
$A\approx^+_\eps B\wrt(E,F)$, if
\[
x^\top(A-B)y\le\eps\sqrt{x^\top Ex}\sqrt{y^\top Fy}
\quad\text{for all }x\in\mathbb R^n_{\ge0},\ y\in\mathbb R^m_{\ge0}.
\]
\end{definition}

The two-sided notion is equivalently characterized by
\[
|x^\top(A-B)y|\le\frac{\eps}{2}(x^\top Ex+y^\top Fy)
\quad\text{for all }x,y.
\]

\subsection{Label Space, Lifting and Compression}
Derandomized squaring enlarges the state space by adding label coordinates.
We introduce the averaging, lifting, and compression operators used to compare
matrices on these spaces of different sizes.

For a finite nonempty set $L$, define its uniform averaging matrix by
$J_L\defeq\bone_L\bone_L^\top/|L|$.
Write $[d]\defeq\{1,\ldots,d\}$ and
$J_d\defeq J_{[d]}$.
Tensor products of matrices are denoted by $\otimes$.

Let $T$ denote the number of levels.
For positive integers $d_1,\ldots,d_T$ to be specified by the construction, let
$L_0\defeq\{1\}$ and
$L_t\defeq L_{t-1}\times[d_t]$ for
$1\le t\le T$.
The corresponding state space at level $t$ is
$V\times L_t$.
\begin{definition} [Uniform lifting and compression] \label{def:lifting}
For $1\le t\le T$, define the lifting matrix from
$\mathbb R^{V\times L_{t-1}}$ to
$\mathbb R^{V\times L_t}$ by
\[
E_t\defeq I_{V\times L_{t-1}}
\otimes\frac{\bone_{d_t}}{\sqrt{d_t}}.
\]
For a matrix $A$ on
$V\times L_{t-1}$ and a matrix
$B$ on
$V\times L_t$, define the uniform lift and compression by
\[
\operatorname{Lift}_t(A)\defeq E_tAE_t^\top,
\qquad
\operatorname{Comp}_t(B)\defeq E_t^\top BE_t.
\]
\end{definition} 

The uniform lift satisfies
\[
\operatorname{Lift}_t(A)=A\otimes J_{d_t}.
\]
Since $E_t^\top E_t=I$, compression inverts the uniform lift:
\[
\operatorname{Comp}_t\bigl(\operatorname{Lift}_t(A)\bigr)=A.
\]

\subsection{Expanders}\label{subsec:expanders}
We use undirected regular spectral expanders, allowing parallel edges and self-loops.
\begin{definition}[$\eps$-spectral expander]\label{def:expander}
An undirected regular multigraph $H$ on
$D$ vertices is an
$\eps$-spectral expander if its random walk matrix
$A_H$ satisfies
\[
\norm{A_H-J_D}_2\le\eps.
\]
\end{definition}

The following form of the small-bias construction of Alon, Goldreich,
H\aa stad, and Peralta~\cite{AGHP92:journal} gives an expander with
power-of-two size and degree, together with a low-space neighbor operation.
We use the following low-space construction for the expander graphs in our derandomized squaring procedure.

\begin{lemma}\label{lem:expander}
For every power of two $D\ge2$ and every
$0<\eps\le1/2$, there is an undirected
$d$-regular
$\eps$-spectral expander on
$[D]$, where
$d$ is a power of two satisfying
\[
d=O\!\left(\left(\frac{\log D}{\eps}\right)^2\right).
\]
The adjacency list of each vertex can be enumerated using
$
O\!\left(\log D+\log(1/\eps)\right)
$
bits of work space.
\end{lemma}
\begin{proof}
The standard Cayley-graph interpretation of Construction~3 and
Proposition~3 in~\cite{AGHP92:journal} gives the claimed degree and spectral
bound.
In that construction, each generator bit is computed by arithmetic over a
binary field whose elements use
$O(\log\log D+\log(1/\eps))$ bits.
This arithmetic, including finding a defining irreducible polynomial by
exhaustive search, uses space linear in that bit length.
Together with the $O(\log D)$ bits needed to store a vertex, this gives the
claimed work-space bound.
\end{proof}

\section{One-Dimensional Chain}\label{sec:one-dim}

In this section, we define the one-dimensional matrix chain that characterizes the stable matrix chain in the two-dimensional case. This gives the approximation and trimming guarantees for the final stable matrices.

Starting from $Q^{(0)}=W$, each level
$1\le t\le T$ consists of the following three steps:
\begin{align}
P^{(t)}
&\defeq \textsc{DerandSquare}(Q^{(t-1)}, H_t)
&&\text{(\Cref{def:derand-square})}
\label{eq:one-dim-square}\\
Z^{(t)}
&\defeq \textsc{Trim}(P^{(t)}, \alpha, \beta_t)
&&\text{(\Cref{def:1D-trimming})}
\label{eq:one-dim-trim}\\
Q^{(t)}
&\defeq P^{(t)}-Z^{(t)}
+\diag\bigl((Z^{(t)})^\top\bone\bigr)
&&
\label{eq:one-dim-assemble}
\end{align}
Equation~\eqref{eq:one-dim-square} defines the raw square
$P^{(t)}$ by applying derandomized squaring
(\Cref{def:derand-square}) to the previous matrix
$Q^{(t-1)}$, where
$H_t$ is the level-$t$ expander on the label space
$[d_t]$.
Equation~\eqref{eq:one-dim-trim} defines the deletion matrix
$Z^{(t)}$ through the trimming process
(\Cref{def:1D-trimming}), with row budget
$\alpha$ and column budget
$\beta_t$.
Equation~\eqref{eq:one-dim-assemble} forms
$Q^{(t)}$ by removing this mass and returning each column's deleted mass to
its diagonal entry, preserving column stochasticity.
Under the assumptions and parameter choices in
Section~\ref{subsec:one-dim-guarantees}, Lemma~\ref{lem:one-dim-invariants}
shows that the resulting chain satisfies the target stable matrix chain
conditions in \Cref{def:target-stable-matrix-chain}.

\subsection{Definitions and Invariants}

At level $t$, derandomized squaring lifts the previous state space by adding
the level-$t$ label space $[d_t]$.
We first define the lifted state space and the ideal lifted walk matrix as
follows.

\begin{definition}[The lifted state space and the ideal lifted walk matrix] \label{def:ideal-lifted-walk}
The lifted state space at level $t > 0$ is defined as
\[
L_t \defeq L_{t-1} \times [d_t],
\]
where $L_0 \defeq \{1\}$ denotes the singleton set and $[d_t] = \{1,\ldots,d_t\}$.
The ideal lifted walk matrix at level $t$ is defined as
\[
X^{(t)} \defeq W^{2^t} \otimes J_{L_t}
\]
Let $s^{(t)}$ denote the stationary distribution of
$X^{(t)}$, i.e., a tensor product of $s(W)$ with the uniform distribution. Let $S^{(t)} \defeq \diag(s^{(t)})$.
\end{definition}

\begin{definition}[The target stable matrix chain] \label{def:target-stable-matrix-chain}
	 A matrix chain 
	 $
	 \{ Q^{(t)} \}_{t = 0}^{T}
	$
	is called \emph{the target stable matrix chain} with parameters $(\Gamma, \eta_0, \dots, \eta_T)$ if $Q^{(t)}$ is a nonnegative square column-stochastic matrix defined on $V \times L_t$ and satisfies the following invariants:
	\begin{itemize}
		\item Full operator error:
		\[
		\bigl\|Q^{(t)}-X^{(t)}\bigr\|_{(S^{(t)})^{-1}}
		\le\Gamma\eta_t,
		\]
		\item Bounded row-mass:
		\[
		Q^{(t)}\mathbf1\le O(\kappa^2)\mathbf1.
		\]
		\item One-sided error:
		\[
		Q^{(t)}
		\approx_{\eta_t}^{+}X^{(t)}
		\ \mathrm{wrt}\ (S^{(t)},(S^{(t)})^{-1})
		\]
\end{itemize}

\end{definition}



We will pick the parameters for the target stable matrix chain as follows:
\[
\eta_t=8^t\eta_0,\qquad
\Gamma=\kappa^{O(1)}(\log n)^{O(1)}.
\]
where $\eta_0$ will be picked sufficiently small that
\[
\Gamma\eta_T\le O(\kappa^{-2}).
\]


The full operator error is the final error bound required by the Richardson
iteration.
The one-sided error controls the \emph{trimming} stage used to keep the
in-degree mass bounded, which we use to address the following main difficulty:
Even though $W^{2^t}$ has the same stationary distribution
$s(W)$ and satisfies
$W^{2^t} \bone \le \kappa^2 \bone$, the
$\ell_2$-approximation does not control row sums.
Derandomized squaring does not work if the row sums (i.e., in-degree mass) blow
up too much.
We therefore use trimming to bound $Q^{(t)} \bone \le O(\kappa^2) \bone$.

\subsection{Derandomized Squaring}

In this section, we describe the derandomized squaring used to construct the target stable matrix chain \Cref{def:target-stable-matrix-chain}.
This derandomized squaring operation is also used in the construction of the two-dimensional matrix chain.

\begin{definition}[Derandomized squaring at level $t$] \label{def:derand-square}
	Let $Q$ be a nonnegative column-stochastic matrix at level $t-1$, i.e., defined on $V \times L_{t-1}$. Let $H$ denote an $\eps$-spectral expander graph defined on level-$t$ label space $[d_t]$. The derandomized squaring of $Q$ with expander $H$ plugged in, denoted by
	\[
  P \gets \textsc{DerandSquare}(Q,H),
  \]
  is defined as follows:
  Let $A_H$ be the random walk matrix of
  $H$, and set
  $
  \widehat{A}_H\defeq I_{V\times L_{t-1}}\otimes A_H.
  $
  Fix an arbitrary nonnegative column-stochastic transition lift matrix $L(Q)$ on
  $V\times L_t$ satisfying
  \[
  \operatorname{Comp}_t(L(Q)) = Q.
  \]
  and we obtain
  \[
  P\defeq\widehat{A}_H L(Q)\widehat{A}_H L(Q)\widehat{A}_H.
  \]
\end{definition}






This definition specifies the matrix operation.
Its sparse low-space realization is discussed in
Section~\ref{sec:implementation}.
We define the derandomized square matrix chain by
\begin{align} \label{eq:P-chain}
P^{(t)} \gets \textsc{DerandSquare}(Q^{(t-1)}, H_t)
\end{align}
for all $t$. Consequently, $P^{(t)}$ is nonnegative and column-stochastic and satisfies the following properties.


\begin{lemma}\label{lem:derand-square}
Let $1\le t\le T$, and suppose
$Q^{(t-1)}$ satisfies the level-$(t-1)$ invariants of
\Cref{def:target-stable-matrix-chain} with
parameters  $(\Gamma, \eta_0, \dots, \eta_T)$.
Assume $\Gamma\ge1$ and $\Gamma\eta_{t-1}\le1$.
Assume that its transition lift satisfies
\[
L(Q^{(t-1)})\bone\le C \bone.
\]
for $C \le O(\kappa^2)$
and the expander $H_t$ is an $\eps_t$-spectral expander with 
$
\eps_t\le\frac{\eta_{t-1}}{3\kappa C}.
$
Then $P^{(t)}$ is nonnegative and column-stochastic, and satisfies the following bounds:
\begin{itemize}
\item\label{itm:derand-square:1} Increased bounded row-mass:
  \[
  P^{(t)}\bone\le C^2\bone=O(\kappa^4)\bone.
  \]
\item\label{itm:derand-square:2} One-sided error:
  \[
  P^{(t)}
  \approx_{4\eta_{t-1}}^{+}X^{(t)}
  \wrt\bigl(S^{(t)},(S^{(t)})^{-1}\bigr).
  \]
\item\label{itm:derand-square:3} Full operator error:
  \[
  \norm{P^{(t)}-X^{(t)}}_{(S^{(t)})^{-1}}
  \le4\Gamma\eta_{t-1}.
  \]
\end{itemize}
\end{lemma}
\begin{proof}
Write $Q=Q^{(t-1)}$,
$X=X^{(t-1)}$,
$S=S^{(t-1)}$, and
$\Pi\defeq E_tE_t^\top$.
By $E_t^\top L(Q) E_t = Q$, we have
\[
\Pi L(Q)\Pi L(Q)\Pi=\operatorname{Lift}_t(Q^2).
\]
Column stochasticity and the row-mass bound give
$\norm{L(Q)}_2\le\sqrt C$.
Replacing the three expander factors by $\Pi$ one at a time, using
$\norm{\widehat{A}_H-\Pi}_2\le\eps_t$ and the fact that both factors have
spectral norm one, gives
\[
\norm{P^{(t)}-\operatorname{Lift}_t(Q^2)}_2\le3C\eps_t.
\]
Translating $\ell_2$-norm into the stationary norm loses at most a factor of
$\kappa$, so
\begin{equation}\label{eq:derand-square-lift-error}
\norm{P^{(t)}-\operatorname{Lift}_t(Q^2)}_{(S^{(t)})^{-1}}
\le3\kappa C\eps_t\le\eta_{t-1}.
\end{equation}
Also, $\norm{X}_{S^{-1}}\le1$ by the property of the stationary distribution, and hence
$\norm{Q}_{S^{-1}}\le1+\Gamma\eta_{t-1}\le2$.
Uniform lifting preserves stationary operator bounds and one-sided
approximation: in stationary coordinates it tensors the matrices with
$J_{d_t}$, a nonnegative orthogonal projection.
Moreover, $\operatorname{Lift}_t(X^2)=X^{(t)}$.
\begin{itemize}
\item \textbf{Row mass.}
  All factors are nonnegative and column-stochastic, and
  $\widehat{A}_H\bone=\bone$.
  Applying the row-mass bound to the two occurrences of
  $L(Q)$ gives
  $P^{(t)}\bone\le C^2\bone$.
\item \textbf{One-sided error.}
  For $x,y\ge0$, expand
  $Q^2-X^2=(Q-X)Q+X(Q-X)$ and apply the input one-sided bound twice:
  \[
  \begin{aligned}
  x^\top(Q^2-X^2)y
  &\le\eta_{t-1}\left(
  \norm{x}_S\norm{Qy}_{S^{-1}}
  +\norm{X^\top x}_S\norm{y}_{S^{-1}}\right)\\
  &\le3\eta_{t-1}\norm{x}_S\norm{y}_{S^{-1}}.
  \end{aligned}
  \]
  For nonnegative vectors $u,v$ on $V\times L_t$, uniform lifting
  preserves this one-sided bound.
  Thus, by Cauchy--Schwarz and
  Equation~\eqref{eq:derand-square-lift-error},
  \[
  \begin{aligned}
  u^\top(P^{(t)}-X^{(t)})v
  &=u^\top\bigl(P^{(t)}-\operatorname{Lift}_t(Q^2)\bigr)v\\
  &\quad+u^\top\operatorname{Lift}_t(Q^2-X^2)v\\
  &\le\left(
    \norm{P^{(t)}-\operatorname{Lift}_t(Q^2)}_{(S^{(t)})^{-1}}
    +3\eta_{t-1}\right)
    \norm{u}_{S^{(t)}}\norm{v}_{(S^{(t)})^{-1}}\\
  &\le4\eta_{t-1}\norm{u}_{S^{(t)}}\norm{v}_{(S^{(t)})^{-1}}.
  \end{aligned}
  \]
\item \textbf{Full operator error.}
  The same expansion gives
  \[
  \norm{Q^2-X^2}_{S^{-1}}
  \le\left(\norm{Q}_{S^{-1}}+\norm{X}_{S^{-1}}\right)
  \norm{Q-X}_{S^{-1}}
  \le3\Gamma\eta_{t-1}.
  \]
  Since $\operatorname{Lift}_t(X^2)=X^{(t)}$ and uniform lifting preserves
  the stationary operator norm, the triangle inequality and
  Equation~\eqref{eq:derand-square-lift-error} give
  \[
  \begin{aligned}
  \norm{P^{(t)}-X^{(t)}}_{(S^{(t)})^{-1}}
  &\le \norm{P^{(t)}-\operatorname{Lift}_t(Q^2)}_{(S^{(t)})^{-1}}\\
  &\quad+\norm{\operatorname{Lift}_t(Q^2-X^2)}_{(S^{(t)})^{-1}}\\
  &\le\eta_{t-1}+3\Gamma\eta_{t-1}\\
  &\le4\Gamma\eta_{t-1}.
  \end{aligned}
  \]
\end{itemize}
\end{proof}

\subsection{Trimming for the One-Dimensional Chain}

In this section, we define the trimming process for the one-dimensional chain.
This process is used only in the analysis.
In the two-dimensional chain, it is implemented through multiple trimming
steps.

\begin{definition}[Trimming process] \label{def:1D-trimming}
Let $P$ be a nonnegative square column-stochastic matrix.
The trimming of $P$ with row budget $\alpha\ge0$ and column budget
$\beta\ge0$, denoted by
\[
Z\gets\textsc{Trim}(P,\alpha,\beta),
\]
returns a deletion matrix $Z$ satisfying
\[
0\le Z\le P,
\qquad
(P-Z)\bone\le\alpha\bone,
\qquad
Z^\top\bone\le\beta\bone,
\]
whenever such a matrix exists.
\end{definition}

The deletion matrix records the mass removed by trimming.
Returning this mass to the diagonal preserves column sums.
The following lemma bounds the resulting row mass and approximation error.

\begin{lemma}\label{lem:trimming}
Let $P$ be a nonnegative column-stochastic matrix on $V\times L_t$ with
$P\bone\le C^2\bone$.
Let $Z\gets\textsc{Trim}(P,\alpha,\beta)$ and define
\[
Q\defeq P-Z+\diag(Z^\top\bone).
\]
Then $Q$ is nonnegative and column-stochastic, and satisfies:
\begin{itemize}
\item\label{itm:trimming:1} Bounded row-mass:
  \[
  Q\bone\le(\alpha+\beta)\bone.
  \]
\item\label{itm:trimming:2} One-sided error:
  \[
  Q\approx_{\beta}^{+}P
  \wrt\bigl(S^{(t)},(S^{(t)})^{-1}\bigr).
  \]
\item\label{itm:trimming:3} Full operator error:
  \[
  \norm{Q-P}_{(S^{(t)})^{-1}}\le\kappa C\sqrt{\beta}+\beta.
  \]
\end{itemize}
\end{lemma}
\begin{proof}
Write $\Delta\defeq\diag(Z^\top\bone)$, so $0\le\Delta\le\beta I$.
\begin{itemize}
\item \textbf{Row mass.}
  Since $P-Z\ge0$, the matrix $Q$ is nonnegative.
  Diagonal compensation preserves every column sum, and
  \[
  Q\bone=(P-Z)\bone+\Delta\bone\le(\alpha+\beta)\bone.
  \]
\item \textbf{One-sided error.}
  For $x,y\ge0$, nonnegativity of $Z$ and Cauchy--Schwarz give
  \[
  x^\top(Q-P)y\le x^\top\Delta y\le\beta x^\top y
  \le\beta\norm{x}_{S^{(t)}}\norm{y}_{(S^{(t)})^{-1}}.
  \]
\item \textbf{Full operator error.}
  The column budget and $0\le Z\le P$ give
  \[
  \norm{Z}_2\le\sqrt{\norm{Z}_1\norm{Z}_\infty}\le C\sqrt{\beta}.
  \]
  Converting to the stationary norm costs at most a factor of $\kappa$.
  Since $\Delta$ is diagonal, its stationary operator norm is at most $\beta$.
  Therefore,
  \[
  \begin{aligned}
  \norm{Q-P}_{(S^{(t)})^{-1}}
  &\le\norm{Z}_{(S^{(t)})^{-1}}+\norm{\Delta}_{(S^{(t)})^{-1}}\\
  &\le\kappa C\sqrt{\beta}+\beta.
  \end{aligned}
  \]
\end{itemize}
\end{proof}

\subsection{Assembling the One-Dimensional Chain}\label{subsec:one-dim-guarantees}

We now combine the squaring and trimming bounds to establish the invariants of
\Cref{def:target-stable-matrix-chain}.
Here we assume that the required deletion matrices exist.
Their existence and construction are addressed in Section~\ref{sec:trimming}.

\paragraph{Parameter choices.}
For $0\le t\le T$, set
\[
\eta_t\defeq8^t\eta_0.
\]
Suppose the transition lifts and expanders satisfy
the row-mass and expander error conditions of Lemma~\ref{lem:derand-square}
with a common
$C=O(\kappa^2)$.
We use trimming threshold
$\alpha=\Theta(\kappa^2)$ and choose
$\Gamma\ge1$ and
$\eta_0>0$ so that
$\Gamma\eta_T\le1$ and, for every
$1\le t\le T$,
\[
\beta_t\le4\eta_{t-1},
\qquad
\kappa C\sqrt{\beta_t}+\beta_t\le4\Gamma\eta_{t-1}.
\]
For the choice $\beta_t=\poly(\kappa,\log n)\eta_{t-1}^2$ used in
Section~\ref{sec:trimming}, the square-root term is a fixed polynomial factor
times $\eta_{t-1}$.
Thus we can satisfy these conditions by first choosing
$\Gamma$ sufficiently large and then
$\eta_0$ sufficiently small.

\begin{lemma}\label{lem:one-dim-invariants}
The chain defined by
Equations~\eqref{eq:one-dim-square}--\eqref{eq:one-dim-assemble} with parameters defined as above 
is a target stable matrix chain (\Cref{def:target-stable-matrix-chain}) with parameters
$(\Gamma,\eta_0,\ldots,\eta_T)$.
\end{lemma}
\begin{proof}
We prove the invariants by induction on $t$.
At level zero, $Q^{(0)}=X^{(0)}=W$ and $W\bone\le\kappa^2\bone$.
This row bound follows from $Ws(W)=s(W)$ and the coordinate bounds on $s(W)$.
For the induction step, apply Lemma~\ref{lem:derand-square} to $Q^{(t-1)}$
and Lemma~\ref{lem:trimming} to $P^{(t)}$.
The latter gives nonnegativity and column stochasticity of $Q^{(t)}$.
\begin{itemize}
\item \textbf{Row mass.}
  The budget conditions give $\beta_t\le4$, and hence
  \[
  Q^{(t)}\bone\le(\alpha+\beta_t)\bone\le O(\kappa^2)\bone.
  \]
\item \textbf{One-sided error.}
  Adding the two one-sided bounds gives
  \[
  Q^{(t)}\approx_{4\eta_{t-1}+\beta_t}^{+}X^{(t)}
  \wrt\bigl(S^{(t)},(S^{(t)})^{-1}\bigr),
  \]
  where $4\eta_{t-1}+\beta_t\le8\eta_{t-1}=\eta_t$.
\item \textbf{Full operator error.}
  By the triangle inequality,
  \[
  \begin{aligned}
  \norm{Q^{(t)}-X^{(t)}}_{(S^{(t)})^{-1}}
  &\le\norm{Q^{(t)}-P^{(t)}}_{(S^{(t)})^{-1}}
  +\norm{P^{(t)}-X^{(t)}}_{(S^{(t)})^{-1}}\\
  &\le\kappa C\sqrt{\beta_t}+\beta_t+4\Gamma\eta_{t-1}\\
  &\le8\Gamma\eta_{t-1}=\Gamma\eta_t.
  \end{aligned}
  \]
\end{itemize}
\end{proof}

\section{The Trimming Step}\label{sec:trimming}

Recall in \Cref{def:1D-trimming}, we defined the trimming process that decomposes the derandomized square matrix $P$ into
\[
P = Z + (P-Z)
\]
such that
\[
0
\le Z
\le P,
\quad
(P-Z) \bone
\le \alpha\bone,
\qquad
Z^\top \bone \le \beta\bone.
\]
This removes at most $\beta$ mass from each column and reduces the row sums to at most $\alpha$.
The final construction uses a partial version of the trimming process, called \emph{the trimming step},
that resolves most of the heavy rows instead of all heavy rows. This allows a simple construction which can be implemented in low space.

\subsection{Existence of a Trimming}

We first characterize the existence of the trimming solution using a bipartite
flow problem. For a matrix $P$, a row subset
$I$, and a column subset
$J$, write
\[
P(I,J)\defeq\sum_{i\in I,j\in J}P_{ij}.
\]

The following lemma gives an analogue of Hall's theorem for the trimming
problem.
We refer to this criterion as the \emph{Hall condition}.
\begin{lemma}\label{lem:trimming-existence}
Let $P$ be a nonnegative matrix and
$\alpha,\beta\ge0$.
There exists a matrix $Z$ satisfying
\[
0\le Z\le P,\qquad
(P-Z)\bone\le\alpha\bone,\qquad
Z^\top\bone\le\beta\bone
\]
if and only if
\[
P(I,J)\le\alpha|I|+\beta|J|
\]
for every row subset $I$ and column subset
$J$.
\end{lemma}
\begin{proof}
\textbf{Necessity.}
  For any row subset $I$ and column subset
  $J$, the row and column budgets give
  \[
  \begin{aligned}
  P(I,J)
  &=(P-Z)(I,J)+Z(I,J)\\
  &\le\sum_{i\in I,j}(P_{ij}-Z_{ij})
  +\sum_{i,j\in J}Z_{ij}\\
  &\le\alpha|I|+\beta|J|.
  \end{aligned}
  \]
\textbf{Sufficiency.}
  Construct a flow network with a source, one vertex per row, one vertex per
  column, and a sink.
  Give the edge from the source to row $i$ capacity
  $\sum_jP_{ij}$, the edge from row
  $i$ to column
  $j$ capacity
  $P_{ij}$, and the edges from each row and column to the sink capacities
  $\alpha$ and
  $\beta$, respectively.
  Write $M=\sum_{i,j}P_{ij}$.
  A cut whose source side contains rows $I$ and columns
  $J$ has capacity
  \[
  \begin{aligned}
  &\sum_{i\notin I,j}P_{ij}
  +\sum_{i\in I,j\notin J}P_{ij}+\alpha|I|+\beta|J|\\
  &\qquad=M-P(I,J)+\alpha|I|+\beta|J|\\
  &\qquad\ge M.
  \end{aligned}
  \]
  The max-flow/min-cut theorem therefore gives a flow of value $M$,
  saturating every edge out of the source.
  Let $Z_{ij}$ be the flow from row
  $i$ to column
  $j$.
  Edge capacities give $0\le Z\le P$.
  Flow conservation shows that the retained row mass flows directly to the
  sink and that each column's deleted mass flows through its sink edge.
  These masses are at most $\alpha$ and
  $\beta$, respectively.
\end{proof}

The one-sided approximation bound implies the Hall condition as follows.
\begin{lemma}\label{lem:one-sided-hall}
Let $P$ be a nonnegative matrix on
$V\times L_t$ satisfying
\[
P\approx^+_\delta X^{(t)}
\wrt\bigl(S^{(t)},(S^{(t)})^{-1}\bigr)
\]
for some $\delta\ge0$.
Then, for all $I,J\subseteq V\times L_t$,
\[
P(I,J)\le2\kappa^2|I|+\frac{\delta^2}{4}|J|.
\]
\end{lemma}
\begin{proof}
Set $x=\bone_I$,
$y=\bone_J$, and
$N=n|L_t|$.
The coordinate bounds on the lifted stationary distribution give
\[
x^\top S^{(t)}x\le\frac{\kappa|I|}{N},
\qquad
y^\top(S^{(t)})^{-1}y\le\kappa N|J|.
\]
Also, since $X^{(t)}\bone\le\kappa^2\bone$,
\[
x^\top X^{(t)}y
\le x^\top X^{(t)}\bone
\le\kappa^2|I|.
\]
Applying the one-sided bound yields
\[
\begin{aligned}
P(I,J)
&\le x^\top X^{(t)}y
+\delta\sqrt{x^\top S^{(t)}x}
\sqrt{y^\top(S^{(t)})^{-1}y}\\
&\le\kappa^2|I|+\delta\kappa\sqrt{|I||J|}\\
&\le2\kappa^2|I|+\frac{\delta^2}{4}|J|.
\end{aligned}
\]
The last inequality follows by expanding
$(\kappa\sqrt{|I|}-(\delta/2)\sqrt{|J|})^2\ge0$.
\end{proof}


\subsection{A Single Trimming Step}\label{subsec:single-trimming-step}

We define a single trimming step and then bound the number of failed rows.

\begin{definition}[A single trimming step]\label{def:one-step-trim}
Let $P$ be a nonnegative matrix,
$I$ a subset of its rows, and
$\alpha,\beta>0$ the row and column budgets.
The trimming step
\[
Z\gets\textsc{OneStepTrim}(P,I,\alpha,\beta)
\]
is defined as follows.
First define the column ratio
\[
\theta_j\defeq\min\left\{1,
\frac{\beta}{\sum_{i\in I}P_{ij}}\right\},
\]
with $\theta_j=1$ when the denominator is zero.
We propose to delete $P_{ij}\theta_j$ mass from the
$(i,j)$-th entry for each
$i\in I$.
We accept the proposal for a row if the remaining row mass is at most
$\alpha$, and set
\[
Z_{ij}\defeq
\begin{cases}
P_{ij}\theta_j
& i\in I,\ \sum_jP_{ij}(1-\theta_j)\le\alpha,\\
0 & \text{otherwise}.
\end{cases}
\]
\end{definition}

The trimming step tries to delete mass from the rows in $I$ while respecting the column budget $\beta$, so it remains to bound the number of failed rows.
The failure set is defined as 
\[
F\defeq\left\{i\in I:
\sum_jP_{ij}(1-\theta_j)>\alpha\right\}.
\]
The following lemma gives the guarantees of a single trimming step.
\begin{lemma}[Trimming step guarantees] \label{lem:one-step-trim}
Let $Z\gets\textsc{OneStepTrim}(P,I,\alpha,\beta)$
Suppose $\alpha_H>0$ and
$\beta_H\ge0$ satisfy the Hall condition
\[
P(I',J')\le\alpha_H|I'|+\beta_H|J'|
\qquad\text{for all row subsets }I'\text{ and column subsets }J'.
\]
If $\alpha\ge2\alpha_H$, then:
\begin{itemize}
\item\label{itm:one-step-trim:1} \textbf{Deletion and row budgets.}
  We have $0\le Z\le P$,
  $Z^\top\bone\le\beta\bone$, and
  \[
  \sum_j(P_{ij}-Z_{ij})\le\alpha
  \qquad\text{for every }i\in I\setminus F.
  \]
\item\label{itm:one-step-trim:2} \textbf{Failed rows.}
  \[
  |F|\le
  \frac{\beta_H\norm{P\bone}_\infty}{\alpha_H\beta}|I|.
  \]
\end{itemize}
\end{lemma}
\begin{proof}
\textbf{Deletion and row budgets.}
Since $0\le\theta_j\le1$, we have
$0\le Z\le P$.
For each column,
\[
\begin{aligned}
\sum_iZ_{ij}
&\le\theta_j\sum_{i\in I}P_{ij}\\
&=\min\left\{\sum_{i\in I}P_{ij},\beta\right\}\\
&\le\beta.
\end{aligned}
\]
Every accepted row satisfies
\[
\sum_j(P_{ij}-Z_{ij})
=\sum_jP_{ij}(1-\theta_j)\le\alpha.
\]

\textbf{Failed rows.}
Define the set of congested columns
\[
C\defeq\left\{j:\sum_{i\in I}P_{ij}\ge\beta\right\}.
\]
Its size satisfies
\[
\begin{aligned}
\beta|C|
&\le\sum_{j\in C}\sum_{i\in I}P_{ij}\\
&\le\sum_{i\in I}(P\bone)_i\\
&\le|I|\norm{P\bone}_\infty.
\end{aligned}
\]
Outside $C$, we have
$\theta_j=1$.
The definition of the failure set and the Hall condition give
\[
\begin{aligned}
\alpha|F|
&\le\sum_{i\in F,j\in C}P_{ij}(1-\theta_j)\\
&\le P(F,C)\\
&\le\alpha_H|F|+\beta_H|C|.
\end{aligned}
\]
Since $\alpha\ge2\alpha_H$, it follows that
\[
\begin{aligned}
|F|
&\le\frac{\beta_H}{\alpha-\alpha_H}|C|\\
&\le\frac{\beta_H}{\alpha_H}|C|\\
&\le\frac{\beta_H\norm{P\bone}_\infty}
{\alpha_H\beta}|I|.
\end{aligned}
\]
\end{proof}


\subsection{Parameters for the Trimming Step}\label{subsec:trimming-parameters}

For the raw square $P^{(t)}$, \Cref{lem:derand-square} gives one-sided error
$4\eta_{t-1}$ and row-mass bound
$P^{(t)}\bone\le C^2\bone$.
Applying \Cref{lem:one-sided-hall} gives Hall budgets
\[
\alpha_H=2\kappa^2,\qquad
\beta_H=4\eta_{t-1}^2.
\]
For our final algorithm, pick $I$ to be the set of heavy rows and choose
\[
\alpha=4\kappa^2,\qquad
\beta\ge\frac{2C^2}{\kappa}\eta_{t-1}^2.
\]
By \Cref{lem:one-step-trim}, the failure set satisfies
\[
\begin{aligned}
|F|
&\le\frac{\beta_H C^2}{\alpha_H\beta}|I|\\
&=\frac{2C^2\eta_{t-1}^2}{\kappa^2\beta}|I|\\
&\le\frac{|I|}{\kappa}.
\end{aligned}
\]
These bounds also apply to any nonnegative residual matrix bounded entrywise
by $P^{(t)}$.
Since $C=O(\kappa^2)$ and
$\eta_t=8\eta_{t-1}$, a sufficiently large constant in the per-step column
budget
\[
\beta=\Theta(\eta_t^2\kappa^3)
\]
ensures this reduction by a factor of $\kappa$.

The total column budget is accounted for in
Section~\ref{sec:two-dim}.

\section{The Two-Dimensional Chain}\label{sec:two-dim}

If we apply derandomized squaring and trimming alternately and complete each
level before proceeding to the next, a level may require
$O(\log_\kappa(n|L_t|))$ trimming steps.
These costs accumulate over $T=O(\log\kappa)$ levels, increasing the recursion
depth and the space required to evaluate the chain.

In this section, we define the 2D chain, which performs one trimming step per
level in each synchronized round.
Our goal is to bound the number of rounds until the chain stabilizes by
$O(\log_\kappa n + \log\kappa)$.
The parameter choices are given in
Section~\ref{subsec:two-dim-guarantees}.
Starting from $Q^{(0,r)}=W$ for every
$r\ge0$, the chain at each level
$1\le t\le T$ and round
$r\ge0$ is defined by
\begin{align}
P^{(t,r)}
&\defeq\textsc{DerandSquare}(Q^{(t-1,r)},H_t)
\label{eq:two-dim-square}\\
Z^{(t,r)}
&\defeq\textsc{Trim}(P^{(t,r)},(P^{(t,r-1)},Z^{(t,r-1)}),\alpha,\beta)
\label{eq:two-dim-trim}\\
Q^{(t,r)}
&\defeq\textsc{Assemble}(P^{(t,r)},Z^{(t,r)})
\label{eq:two-dim-assemble}
\end{align}
Equation~\eqref{eq:two-dim-square} applies derandomized squaring
(\Cref{def:derand-square}) to the preceding level at the current round,
where $H_t$ is the level-$t$ expander.
Equation~\eqref{eq:two-dim-trim} reuses deletion information from the previous
round and applies a single trimming step to the remaining heavy rows, as
specified in \Cref{def:two-dim-trim}.
Equation~\eqref{eq:two-dim-assemble} deletes unresolved heavy rows to bound
the row mass, as specified in \Cref{def:two-dim-assemble}.
It restores each column sum to one by adding mass to its diagonal entry.

Section~\ref{subsec:two-dim-assembly} defines the trimming and assemble steps.
Section~\ref{subsec:two-dim-analysis} bounds the number of rounds until
stabilization.
Section~\ref{subsec:two-dim-guarantees} collects the parameter choices and
proves that the final chain satisfies the target invariants.

\subsection{The Trimming and Assemble Steps}\label{subsec:two-dim-assembly}

In this section, we define the trimming and assemble steps for the two-dimensional chain. We refer to $P^{(t, r)}$ as the raw matrix, $Z^{(t, r)}$ as the deletion matrix, and
$Q^{(t, r)}$ as the trimmed matrix.

The trimming step uses the previous raw matrix and deletion matrix
$(P^{(t,r-1)},Z^{(t,r-1)})$ as reference. It keeps the deletion information if the raw row has not changed compared to the previous round.


\begin{definition}[Two-dimensional trimming step]\label{def:two-dim-trim}
At level $t$, fix row budget
$\alpha>0$ and column budget
$\beta>0$.
For round $r\ge1$, define the filtered deletion matrix by
\[
\Zhat^{(t, r)}_{ij} \defeq
\begin{cases}
Z^{(t, r-1)}_{ij} & \text{if } (P^{(t, r)})^\top \bone_i
= (P^{(t, r-1)})^\top \bone_i, \\
0 & \text{otherwise}.
\end{cases}
\]
At round zero, set $\Zhat^{(t,0)}=0$.
Define the active set of heavy rows by
\[
I \defeq
\Big\{
i  : \norm{(P^{(t, r)} - \Zhat^{(t, r)})^\top \bone_i}_1 > \alpha
\Big\}.
\]
The two-dimensional trimming step returns
\[
Z^{(t, r)}
\gets \Zhat^{(t, r)}
+ \textsc{OneStepTrim}(P^{(t, r)}, I, \alpha, \beta),
\]
using the one-step routine from \Cref{def:one-step-trim}.
\end{definition}

The assemble step then deletes the unresolved heavy rows.
It restores each column sum to one by adding mass to its diagonal entry.

\begin{definition}[Assemble step]\label{def:two-dim-assemble}
Given the raw matrix $P^{(t,r)}$ and deletion matrix
$Z^{(t,r)}$, define
\[
R_{ij}^{(t,r)} \defeq
\begin{cases}
P^{(t, r)}_{ij} - Z^{(t, r)}_{ij} & \text{if }
\sum_{j} (P^{(t, r)}_{ij} - Z^{(t, r)}_{ij}) \le \alpha, \\
0 & \text{otherwise}.
\end{cases}
\]
The assemble step returns
\[
Q^{(t,r)} \defeq R^{(t,r)}
+ \diag\Bigl(\bigl(P^{(t,r)}-R^{(t,r)}\bigr)^\top\bone\Bigr).
\]
\end{definition}


\subsection{Analysis of the Synchronized Process}\label{subsec:two-dim-analysis}

In this section, we give the core round analysis of the two-dimensional chain. We first define the following stabilization round.

\begin{definition}[Stabilization round]\label{def:stabilization-round}
The stabilization round $r_{\max}$ is the smallest nonnegative integer
$r$ such that, for every level
$1\le t\le T$ and every
$\rhat\ge r$,
\begin{align*}
P^{(t,\rhat)} &= P^{(t,r)},\\
Q^{(t,\rhat)} &= Q^{(t,r)},\\
Z^{(t,\rhat)} &= Z^{(t,r)}.
\end{align*}
If no such round exists, set $r_{\max}=\infty$.
\end{definition}

The assemble step gives up on the set of unresolved heavy rows, so it is important to let the chain stabilize to obtain the final approximation guarantee. We will show that when the two-dimensional chain stabilizes, the final stable matrix chain satisfies the properties of the target stable matrix chain in \Cref{def:target-stable-matrix-chain}.  

\begin{lemma}[Stabilization bound]\label{lem:stabilization-rounds}
The two-dimensional chain (defined by \Cref{eq:two-dim-square}--\Cref{eq:two-dim-assemble}) has the stabilization round
\[
r_{\max}=O(\log_\kappa n+\log\kappa).
\]
\end{lemma}

We prove \Cref{lem:stabilization-rounds} by bounding the number of unstable
rows.

\begin{definition}[Unstable rows]\label{def:unstable-rows}
At a fixed level $1\le t\le T$, a row of
$P^{(t,r)}$ or
$Z^{(t,r)}$ is stable at round
$r$ if that row remains unchanged at every later round.
The row is called unstable otherwise.
Define
\begin{align*}
U(t,r) &\defeq \bigl\{i : \text{row $i$ of $P^{(t,r)}$ is unstable}
\lor \text{row $i$ of $Z^{(t,r)}$ is unstable}\bigr\},\\
u(t,r) &\defeq |U(t,r)|.
\end{align*}
At level zero, set $U(0,r)=\varnothing$ and
$u(0,r)=0$ for every round.
\end{definition}

It suffices to find a round $r=O(\log_\kappa n+\log\kappa)$ such that
$u(t,r)=0$ at every level.
Then all raw and deletion matrices remain unchanged at every later round,
so the assemble step also produces the same trimmed matrices.
Therefore, $r_{\max}\le r$.

We first prove that the chain stabilizes after finitely many rounds.
This allows us to compare the intermediate matrices with the final stable
matrices when deriving the sharper bound.

\begin{lemma}[Conservative stabilization bound]\label{lem:conservative-stabilization}
With the parameter choices in Section~\ref{subsec:two-dim-guarantees},
\[
r_{\max}=O(\log n\log\kappa).
\]
\end{lemma}
\begin{proof}
We induct on the level, proving that all levels up to $t$ stabilize within
$O(t\log n)$ rounds with valid trimmings and the target invariants.
First, a retained nonzero deletion row already has residual mass at most
$\alpha$, so it is excluded from the active set.
Retained and newly accepted deletion rows therefore have disjoint supports.
Induction on rounds using \Cref{def:two-dim-trim,def:one-step-trim} gives
\[
\begin{aligned}
0&\le Z^{(t,r)}\le P^{(t,r)},\\
\bigl(Z^{(t,r)}\bigr)^\top\bone&\le(r+1)\beta\bone.
\end{aligned}
\]

The base case is $Q^{(0,r)}=W$ for every round.
Suppose the preceding levels have stabilized with the target invariants.
The raw matrix at level $t$ is then fixed, and
\Cref{lem:derand-square,lem:one-sided-hall} give its row-mass bound and Hall
condition.
With the stated parameter choices, \Cref{lem:one-step-trim} reduces the
number of heavy rows by a factor of $\kappa$ per round.
Successful deletion rows remain valid because the raw matrix no longer
changes.
Starting with at most $n|L_t|$ heavy rows, another
$O(\log_\kappa(n|L_t|))=O(\log n)$ rounds reduce their number below one.
No heavy rows remain, so the deletion matrix and the assembled matrix also
remain fixed thereafter.

This level finishes within $O(t\log n)$ rounds.
The column bound above and the choice of the total budget give
\[
\bigl(Z^{(t,r)}\bigr)^\top\bone
\le O(t\log n)\beta\bone
\le\beta_t\bone
\]
at completion, and every residual row has mass at most $\alpha$.
Thus the final deletion matrix is a valid trimming, and the assemble step
has the one-dimensional form.
The inductive argument in \Cref{lem:one-dim-invariants} gives the target
invariants at this level, completing the induction.
Since $T=O(\log\kappa)$, we conclude that
\[
r_{\max}=O(T\log n)=O(\log n\log\kappa).
\]
\end{proof}

To obtain the sharper bound, we establish the following recurrence for
$1\le t\le T$ and
$r\ge1$:
\begin{equation}\label{eq:unstable-rows-recurrence}
u(t,r)\le\poly(\kappa,\log n)\,u(t-1,r)+\kappa^{-1}u(t,r-1).
\end{equation}
The initial condition is
\[
u(t,0)\le n|L_t|\le n\,O(\kappa\log n)^{O(t)}.
\]

We bound $u(t,r)$ by considering whether row
$i$ of
$P^{(t,r)}$ is fully stable, as defined below.

Row $i$ of
$P^{(t,r)}$ is fully stable if both of the following conditions hold:
\begin{enumerate}
\item Row $i$ of
  $P^{(t,r)}$ is stable.
\item For every column $j$ with
  $P^{(t,r)}_{ij}>0$ and every row
  $w$, one of the following holds:
  \begin{enumerate}
  \item $P^{(t,r)}_{wj}>0$, and row
    $w$ of
    $P^{(t,r)}$ is stable.
  \item $P^{(t,r)}_{wj}=0$, and
    $P^{(t,\rhat)}_{wj}=0$ for every
    $\rhat\ge r$.
  \end{enumerate}
\end{enumerate}

\noindent\textbf{Case 1: Row $i$ of
$P^{(t,r)}$ is not fully stable.}
We prove the bound for this case by
\begin{align*}
&\bigl|\{i:\text{row $i$ of $P^{(t,r)}$ is not fully stable}\}\bigr|\\
&\qquad\le \poly(\kappa,\log n)\,
\bigl|\{w:\text{row $w$ of $P^{(t,r)}$ is unstable}\}\bigr|\\
&\qquad\le \poly(\kappa,\log n)\,
\left|\{w:
\text{row $w$ of $Q^{(t-1,r)}$ is unstable}
\}\right|\\
&\qquad \le \poly(\kappa,\log n)\,u(t-1,r).
\end{align*}
The first inequality follows from the definition of fully stable rows and the
fact that every row and every column in the 2D chain has at most
$\poly(\kappa,\log n)$ nonzeros.
The second inequality uses the locality of derandomized squaring and the
$\poly(\kappa,\log n)$ bound on the degree of expanders.
For the last inequality, each index in $U(t-1,r)$ can affect its own row of
$Q^{(t-1,r)}$ and, through diagonal compensation, the rows indexed by columns
touched by its row of $P^{(t-1,\rhat)}$ at some future round
$\rhat\ge r$.
Thus the last inequality follows if each row touches at most
$\poly(\kappa,\log n)$ distinct columns over these rounds.

\vspace{0.3em}
\noindent\textbf{Case 2: Row $i$ of
$P^{(t,r)}$ is fully stable.}
It remains to bound the indices $i\in U(t,r)$ for which row
$i$ of
$P^{(t,r)}$ is fully stable, so row
$i$ of
$Z^{(t,r)}$ must be unstable.
We bound this case using the one-step trimming guarantee.
The challenge is that the guarantee does not apply to intermediate matrices,
because the one-sided error invariant is available only for the final raw
matrix $P^{(t,r_{\max})}$.

Recall that $Z^{(t,r)}$ is produced by the one-step trimming process
\[
Z^{(t,r)} = \Zhat^{(t,r)}
+ \textsc{OneStepTrim}(P^{(t,r)}, I, \alpha, \beta).
\]
where $I$ denotes the active set used at level
$t$ and round
$r$.
For the same active set $I$, a fully stable raw row has identical trimming
proposals in the current and final raw matrices, since that row and every
column it touches are unchanged.
If its deletion row is unstable, it must be active and fail the current
step, since an inactive or successfully trimmed row would remain unchanged.
The proof of \Cref{lem:conservative-stabilization} supplies the row-mass
bound and Hall condition for the final raw matrix.
The Hall condition holds for every row subset, including the current
active set.
Applying the one-step trimming guarantee to $P^{(t,r_{\max})}$ with the same
active set $I$, we conclude that at most
$|I|/\kappa$ rows whose rows in
$P^{(t,r)}$ are fully stable remain unresolved after this trimming step.
Moreover, $I\subseteq U(t,r-1)$.
Otherwise, an active row whose raw and deletion rows were already stable
would remain heavy in the final matrix, contrary to the successful
completion established by \Cref{lem:conservative-stabilization}.
Therefore,
\[
\bigl|\{i\in U(t,r):
\text{row $i$ of $P^{(t,r)}$ is fully stable}\}\bigr|
\le \frac{|I|}{\kappa}
\le \frac{u(t,r-1)}{\kappa}.
\]

Combining the two cases gives Equation~\eqref{eq:unstable-rows-recurrence}.
The guarantees of the final stable chain are collected in
Section~\ref{subsec:two-dim-guarantees}.

\begin{proof}[Proof of \Cref{lem:stabilization-rounds} from the recurrence]
Choose $K=\poly(\kappa,\log n)\ge1$ to bound the coefficient in
Equation~\eqref{eq:unstable-rows-recurrence} and satisfy
$|L_t|\le K^t$ for every level.
Starting from $u(t,0)\le nK^t$ and
$u(0,r)=0$, induction on
$r$ and then on
$t$ gives
\[
u(t,r)\le n(4K)^t\kappa^{-r/2},
\]
since the two terms of the recurrence contribute factors
$1/4$ and
$\kappa^{-1/2}$, whose sum is at most one for
$\kappa\ge2$.
This bound is less than one at every level when
\[
r>2\log_\kappa n+2T\log_\kappa(4K).
\]
The counts are nonnegative integers, so all rows are stable by this round.
Using $T=O(\log\kappa)$, we obtain
\[
\begin{aligned}
r_{\max}
&\le O\left(\log_\kappa n+T\log_\kappa(4K)\right)\\
&\le O\left(\log_\kappa n+\log\kappa+\log\log n\right)\\
&\le O\left(\log_\kappa n+\log\kappa\right).
\end{aligned}
\]
\end{proof}

\subsection{Guarantees of the Final Stable Chain}\label{subsec:two-dim-guarantees}

We collect the parameter choices and prove that the final chain satisfies
\Cref{def:target-stable-matrix-chain}.
The conservative completion argument in
\Cref{lem:conservative-stabilization} supplies the valid final trimmings
needed for the one-dimensional chain analysis.

\paragraph{Parameter choices.}
Choose $T=O(\log\kappa)$ with
$2^T\ge\kappa$, and use power-of-two label-space sizes
$d_t\le\poly(\kappa,\log n)$, so that
$|L_t|\le\poly(\kappa,\log n)^t$.
Fix the transition lifts and expanders at each level, with the lift row-mass
bound in \Cref{lem:derand-square} for a common
$C=O(\kappa^2)$.
Set $\eta_t=8^t\eta_0$ for
$0\le t\le T$ and choose each level-$t$ expander using
\Cref{lem:expander}, with error
\[
\eps_t\le\frac{\eta_{t-1}}{3\kappa C}
\qquad(1\le t\le T).
\]
As in Section~\ref{subsec:trimming-parameters}, choose
\[
\alpha=4\kappa^2,
\qquad
\beta\ge\frac{2C^2}{\kappa}\eta_{t-1}^2,
\qquad
\beta=\Theta(\eta_t^2\kappa^3).
\]
Here $\beta$ is the per-step column budget at level
$t$, fixed across rounds.
Choose the total column budget
\[
\beta_t=\Theta(\log n\log\kappa)\,\beta
=\poly(\kappa,\log n)\eta_{t-1}^2,
\]
with a sufficiently large constant to cover all rounds in
\Cref{lem:conservative-stabilization}, including round zero.
Choose $\Gamma\ge1$ and
$\eta_0>0$ as in Section~\ref{subsec:one-dim-guarantees}, so that
\[
\Gamma\eta_T\le1,
\qquad
\beta_t\le4\eta_{t-1},
\qquad
\kappa C\sqrt{\beta_t}+\beta_t\le4\Gamma\eta_{t-1}
\]
for every $1\le t\le T$.

\begin{lemma}[Final stable chain]\label{lem:two-dim-invariants}
With the parameter choices above, the final chain
\[
\left\{Q^{(t,r_{\max})}\right\}_{t=0}^{T}
\]
is a target stable matrix chain with parameters
$(\Gamma,\eta_0,\ldots,\eta_T)$ in the sense of
\Cref{def:target-stable-matrix-chain}.
\end{lemma}
\begin{proof}
The level-by-level argument in \Cref{lem:conservative-stabilization} shows that every
level stabilizes with a valid deletion matrix satisfying
\[
\begin{aligned}
0&\le Z^{(t,r_{\max})}\le P^{(t,r_{\max})},\\
\bigl(P^{(t,r_{\max})}-Z^{(t,r_{\max})}\bigr)\bone
&\le\alpha\bone,\\
\bigl(Z^{(t,r_{\max})}\bigr)^\top\bone&\le\beta_t\bone.
\end{aligned}
\]
No unresolved heavy rows remain, so the assemble step gives
\[
Q^{(t,r_{\max})}=P^{(t,r_{\max})}-Z^{(t,r_{\max})}
+\diag\left(\bigl(Z^{(t,r_{\max})}\bigr)^\top\bone\right).
\]
Together with the squaring relation, this is a one-dimensional chain with
valid trimmings, and \Cref{lem:one-dim-invariants} gives the claim.
\end{proof}

\section{Low-Space Implementation}\label{sec:implementation}

Let $h \defeq O(\log\kappa+\log\log n)$ bound the length of a local label for
one level.
We aim to establish the following entry-query bound for a finite-precision
implementation of the 2D chain:
Given a level $t\le T$, a round
$r\le r_{\max}$, and two lifted vertex indices
$i,j \in V \times L_t$, the
$(i,j)$-th entry of
$P^{(t,r)}$,
$Q^{(t,r)}$, or
$Z^{(t,r)}$ can be computed using
$O(\log n+(T+r)h)$ workspace.
In particular, a query at round $r_{\max}$ uses
$O\bigl((\log n+\log^2\kappa)\log\log n\bigr)$ space.

We aim for a Reingold-style implementation based on deterministic enumeration of
labeled walks, maintaining the current vertex address as we follow each walk.

\subsection{Uniform Weighted Implementation of Trimming}\label{subsec:finite-precision}
We replace weighted edges with parallel copies of equal mass, using a short
\textbf{copy index} to distinguish copies of the same edge.
Unequal transition weights are encoded by different numbers of copies.
We round the weights using $O(\log\kappa+\log\log n)$ bits of precision.
Each trimming step rounds the remaining weight of each edge downward and retains
the corresponding number of copies.
The copies accounting for deleted mass and the additional rounding loss are
redirected to self-loops at the source, including the mass removed from
unresolved rows in the assemble step.
Thus every source has the same total number of outgoing copies, and the
resulting transition matrix has column sums exactly one.
The copies are represented implicitly, with their counts computed as needed
rather than stored explicitly.
Each parallel edge is identified by its original edge label and copy index,
using $O(\log\kappa+\log\log n)$ bits in total, the same order as for
derandomized squaring.
The copy index is stored separately from the vertex address.
The rounding-error analysis is deferred to the full proof.

\subsection{Recursive Evaluation}\label{subsec:recursive-evaluation}
We evaluate the 2D chain recursively by enumerating edge labels and, for
parallel edges, their copy indices, recomputing multiplicities and trimming
decisions as needed.
The key is to share the full vertex address across recursive calls.
A lifted vertex consists of an original vertex and the label coordinates
introduced by derandomized squaring, requiring $O(\log n+Th)$ bits in total
where we recall that $T = O(\log \kappa)$ is the number of levels.
We use the following shared register:
\begin{itemize}
\item \texttt{addr} is a mutable register of $O(\log n+Th)$ bits holding the
  current lifted vertex address.
  All recursive calls share \texttt{addr} instead of storing a separate full
  address at each recursive depth.
\end{itemize}

We distinguish two types of recursive calls:
\begin{itemize}
\item \textbf{Query calls} compute a value, such as a row sum or a trimming
  decision.
  They may modify \texttt{addr} temporarily, but restore it to its value on
  entry before returning.
\item \textbf{Navigation calls} use a short edge label, including a copy index
  for parallel edges, to move \texttt{addr} to a neighboring vertex and return a
  short reverse label.
  They leave \texttt{addr} at the destination, and the reverse label identifies
  the same edge copy and allows the caller to return to the original vertex.
\end{itemize}
A query can therefore navigate to a neighbor, make a recursive query there, and
navigate back using the saved reverse label.
For example, to compute the active-set column load
$\sum_{i\in I}P^{(t,r)}_{ij}$, we start with \texttt{addr} at vertex
$j$ and enumerate its outgoing labels one at a time.
For each label, we navigate to its target $i$, query whether
$i\in I$, and return to
$j$ using the saved reverse label.
We add the edge's weight to the load if the target is active.
Thus, the recursive frame retains only a constant number of short labels, the
loop position, and the partial sum used to compute the query result.

\subsection{Space Accounting}\label{subsec:space-accounting}
We account for the target space bound in terms of the workspace required
by the recursive operations.
Squaring uses the preceding level $(t-1,r)$, while retained trimming information
comes from the preceding round $(t,r-1)$.
The required bound is that each recursive frame at level $t>0$ uses only
$O(h)$ bits in addition to the shared vertex address.
Query calls restore this address before returning, and navigation calls return a
short reverse label.
The frame stores only a constant number of edge labels and copy indices, loop
counters, and partial sums.
The intended workspace bounds for the individual operations are as follows:
\begin{itemize}
\item \textbf{Derandomized squaring.}
  We use the standard derandomized squaring routine.
  Navigation composes a constant number of expander steps and navigation calls
  at the preceding level.
  We keep only the short labels needed for the reverse step.
\item \textbf{Trimming.}
  We scan the incident edges one at a time, computing deletion amounts and copy
  counts as needed.
  The scan reuses storage between edges, so each recursive frame keeps only
  short labels, copy indices, and partial sums, using $O(h)$ bits.
\item \textbf{Comparing rows between rounds.}
  To check whether a row of $P^{(t,r)}$ differs from its row in
  $P^{(t,r-1)}$, we scan incoming edges in both rounds and compare the total
  weight from each source.
  Since we cannot store the whole address in the frame, we compare them bit by
  bit so the frame keeps only short navigation labels, a bit position, and a
  comparison bit for these address comparisons.
\item \textbf{Arithmetic.}
  On a common precision of $O(h)$ bits, summing
  $2^{O(h)}$ contributions requires only
  $O(h)$ bits for the accumulator.
  Loop counters, copy indices, and acceptance times also fit in this bound.
\end{itemize}

The following lemma states the space bound under these workspace and
recursion requirements.

\begin{lemma}[Space accounting]\label{lem:implementation-space}
Suppose an entry query at level $t$ and round
$r$ has a recursive evaluation in which each recursive subquery decreases
$t+r$, with only a constant number of intervening helper frames.
Assume that all calls share an $O(\log n+Th)$-bit vertex address, each
suspended frame uses $O(h)$ additional bits, and leaf routines use
$O(\log n+Th)$ workspace.
Then the query uses
\[
O\bigl(\log n+(T+r)h\bigr)
\]
bits of work space.
\end{lemma}
\begin{proof}
Along any recursive call path, the sum of the level and round decreases at
each recursive subquery.
There are therefore $O(t+r+1)$ active frames, using
$O((t+r+1)h)$ bits in total.
The shared address and the workspace of the current leaf are counted once.
Since $t\le T$ and
$T\ge1$, the total space is
\[
O(\log n+Th)+O((t+r+1)h)
=O\bigl(\log n+(T+r)h\bigr).
\]
\end{proof}

Under these assumptions, using \Cref{lem:stabilization-rounds}, the resulting
space bound is
\[
\begin{aligned}
O\bigl(\log n+(T+r_{\max})h\bigr)
&=O\!\left(\log n+
  (\log\kappa+\log_\kappa n)(\log\kappa+\log\log n)\right)\\
&=O\!\left((\log n+\log^2\kappa)\log\log n\right).
\end{aligned}
\]

\section{Wrapping Up: Proofs of the Main Results}\label{sec:main-proofs}

We apply Richardson iteration to approximate the stationary distribution,
using the final target stable matrix chain to construct a preconditioner.
The preconditioner makes the iteration contract toward the stationary vector.
We then obtain the walk-probability guarantee by reducing to the Eulerian case.

\subsection{Approximating Stationary Distribution by Richardson}\label{subsec:stationary-recovery}

Choose the number of levels so that $\kappa\le2^T<2\kappa$ and take
$\eta_0$ small enough that, in addition to the preceding budget conditions,
\[
\Gamma\eta_T\le\frac{1}{16\kappa^2}.
\]
Compress the final matrices to the original vertex set by defining
\begin{equation}\label{eq:compressed-final-chain}
\overline{Q}^{(t)}\defeq
\bigl(\operatorname{Comp}_1\circ\cdots\circ
\operatorname{Comp}_t\bigr)\bigl(Q^{(t,r_{\max})}\bigr),
\end{equation}
with $\overline{Q}^{(0)}=W$.
Uniform compression preserves column stochasticity and does not increase the
operator error in the corresponding stationary norms.
Thus \Cref{lem:two-dim-invariants} gives
\[
\norm{\overline{Q}^{(t)}-W^{2^t}}_{S(W)^{-1}}
\le\Gamma\eta_t.
\]
Define the preconditioner and the iteration matrix by
\begin{equation}\label{eq:recovery-operators}
B\defeq\prod_{t=0}^{T-1}(I+\overline{Q}^{(t)}),
\qquad
M\defeq I-B(I-W).
\end{equation}
The product is ordered by increasing level from left to right.
The preconditioned Richardson iteration is
\begin{equation}\label{eq:richardson-iteration}
x^{(0)}\defeq\frac1n\bone,
\qquad
x^{(q+1)}\defeq x^{(q)}-B(I-W)x^{(q)}=Mx^{(q)}.
\end{equation}
Using the original residual $I-W$ keeps the true stationary vector fixed.

\begin{lemma}[Richardson iteration]\label{lem:stationary-recovery}
The matrix $M$ satisfies
$Ms(W)=s(W)$,
$\bone^\top M=\bone^\top$, and
$\norm{M}_{S(W)^{-1}}\le1$.
For every integer $q\ge0$,
\[
\max_i
\frac{\left|\left(M^q\frac1n\bone\right)_i-s(W)_i\right|}{s(W)_i}
\le\kappa\sqrt n\,2^{-q}.
\]
\end{lemma}
\begin{proof}
The exact product is the short-walk sum
\[
B_*\defeq\prod_{t=0}^{T-1}(I+W^{2^t})
=\sum_{j=0}^{2^T-1}W^j.
\]
Normalize each product factor by two and telescope the difference of the
ordered products.
Since $\norm{W^j}_{S(W)^{-1}}\le1$ and
$\sum_{t<T}\Gamma\eta_t<\Gamma\eta_T/7\le1$, this gives
\[
\begin{aligned}
\norm{B-B_*}_{S(W)^{-1}}
&\le2^T\left(\prod_{t<T}(1+\Gamma\eta_t/2)-1\right)\\
&\le2^T\sum_{t<T}\Gamma\eta_t\\
&\le\frac{2^T\Gamma\eta_T}{7}\\
&\le\frac18.
\end{aligned}
\]
The mixing assumption and the stationary probability bounds imply
\[
\begin{aligned}
\norm{W^\kappa-s(W)\bone^\top}_{S(W)^{-1}}
&\le\kappa n\max_{i,j}
  |(W^\kappa-s(W)\bone^\top)_{ij}|\\
&\le\kappa n^{-9}\le\frac14.
\end{aligned}
\]
The same bound holds with $\kappa$ replaced by
$2^T$, since multiplication by a power of
$W$ is a contraction in the stationary norm.
Using $B_*(I-W)=I-W^{2^T}$, we obtain
\[
M=W^{2^T}-(B-B_*)(I-W).
\]
Consequently, for every vector $x$ with
$\bone^\top x=0$,
\[
\begin{aligned}
\norm{Mx}_{S(W)^{-1}}
&\le\left(\frac14+
  \norm{B-B_*}_{S(W)^{-1}}\norm{I-W}_{S(W)^{-1}}\right)
  \norm{x}_{S(W)^{-1}}\\
&\le\frac12\norm{x}_{S(W)^{-1}}.
\end{aligned}
\]
Since $\bone^\top B=2^T\bone^\top$ and
$(I-W)s(W)=0$, the iteration preserves total mass and fixes
$s(W)$.
The stationary vector is orthogonal to the zero-sum subspace in the stationary
norm, so the contraction above also gives
$\norm{M}_{S(W)^{-1}}\le1$.

Finally,
\[
\norm{\frac1n\bone-s(W)}_{S(W)^{-1}}^2
=\frac1{n^2}\sum_i\frac1{s(W)_i}-1
\le\kappa.
\]
Coordinatewise conversion and contraction give
\[
\begin{aligned}
\max_i
\frac{\left|\left(M^q\frac1n\bone\right)_i-s(W)_i\right|}{s(W)_i}
&\le\sqrt{\kappa n}\,
\norm{M^q\left(\frac1n\bone-s(W)\right)}_{S(W)^{-1}}\\
&\le\kappa\sqrt n\,2^{-q}.
\end{aligned}
\]
\end{proof}

\begin{proof}[Proof of \Cref{thm:main}]
By \Cref{lem:stationary-recovery}, a power of two
$q=O(\log n)$ Richardson steps gives relative error at most
$n^{-c}/2$.
Evaluate each coordinate of $M^q\frac1n\bone$ to additive error
$n^{-c}/(2\kappa n)$, contributing at most another
$n^{-c}/2$ relative error.

Use balanced matrix multiplication for $M^q$ and for the product of
normalized factors $(I+\overline{Q}^{(t)})/2$ defining
$2^{-T}B$.
Entries of $M$ are obtained from
\[
M_{ij}=\mathbf1_{i=j}-B_{ij}+\sum_v B_{iv}W_{vj}.
\]
The two recursions have total depth $O(\log\log n)$.
Every exact normalized interval product has stationary norm at most two,
and every power of $M$ has norm at most one.
Thus each multiplication amplifies inherited norm error by only a constant
factor, while scalar rounding contributes at most
$\poly(n)\nu$ for rounding unit
$\nu$.
Including the final scaling by $2^T<2\kappa$, choosing
$\nu=2^{-O(\log n)}$ therefore gives the required accuracy.
Intermediate values and partial sums have polynomial magnitude.
The two recursions consequently use $O(\log n\log\log n)$ workspace.

Under the entry-query requirements of Section~\ref{sec:implementation},
a compressed entry is computed by summing lifted entries and dividing by
$|L_t|$.
The entry routine and the $O(\log n+Th)$ bits for compression counters and
accumulators are needed only at the leaves, so their workspace is added once.
The total space is
\[
O\!\left((\log n+\log^2\kappa)\log\log n\right).
\]
\end{proof}

\subsection{Approximating Random-Walk Probabilities}\label{subsec:walk-recovery}

We balance the edge weights obtained from the approximate stationary vector
and apply the Eulerian walk-estimation algorithm of
Ahmadinejad et al.~\cite{AKMPSV20}.

\begin{proof}[Proof of \Cref{cor:intro-walk}]
Fix the walk length $k\le\poly(n)$ and requested error
$n^{-c}$.
Use Theorem~\ref{thm:main} with sufficiently high accuracy to obtain a
positive vector $\widehat{s}$ on a common binary grid satisfying
\[
|\widehat{s}_i-s(W)_i|\le\delta s(W)_i,
\qquad
\delta\le\min\left\{\frac12,\frac{n^{-c}}{16k}\right\}.
\]
Define the imbalance vector and its positive and negative parts by
\[
b\defeq W\widehat{s}-\widehat{s},
\qquad
b_i^+\defeq\max\{b_i,0\},
\qquad
b_i^-\defeq\max\{-b_i,0\}.
\]
Since $W$ is column stochastic,
$\sum_i b_i^+=\sum_i b_i^-$.
Write $\sigma$ for this common sum and define the weighted adjacency matrix
\[
F\defeq W\diag(\widehat{s})+
\frac{b^-(b^+)^\top}{\sigma},
\]
where the second term is zero when $\sigma=0$.
The correction makes the row sums and column sums coincide:
\[
F\bone=W\widehat{s}+b^-
=\widehat{s}+b^+=F^\top\bone.
\]
Thus $F$ is Eulerian, with transition matrix
\[
\widetilde W\defeq F\diag(\widehat{s}+b^+)^{-1}.
\]

Nonnegativity of $W$ and stationarity of
$s(W)$ give
\[
|b_i|
\le|(W(\widehat{s}-s(W)))_i|
  +|\widehat{s}_i-s(W)_i|
\le2\delta s(W)_i.
\]
Comparing the original and added mass in each column gives
\[
\begin{aligned}
\norm{\widetilde W-W}_1
&\le\max_j\frac{2b_j^+}{\widehat{s}_j+b_j^+}\\
&\le\frac{4\delta}{1-\delta}\\
&\le8\delta.
\end{aligned}
\]
Telescoping the powers of the two column-stochastic matrices now gives
\[
\begin{aligned}
\norm{\widetilde W^k-W^k}_1
&\le\sum_{a=0}^{k-1}
\norm{\widetilde W^a}_1
\norm{\widetilde W-W}_1
\norm{W^{k-1-a}}_1\\
&\le8k\delta\\
&\le\frac{n^{-c}}2.
\end{aligned}
\]

The common binary grid for $\widehat{s}$ and the constant out-degree of
$W$ give a common polynomial-size denominator for
$F$.
Scaling by this denominator yields a polynomial-size Eulerian multigraph,
whose adjacency lists are generated by recomputing stationary estimates with
$O(\log n)$ additional workspace.
Apply the Eulerian walk-estimation algorithm of
Ahmadinejad et al.~\cite{AKMPSV20} to estimate
$(\widetilde W^k)_{ij}$ to additive error
$n^{-c}/2$ in
$O(\log n\log\log n)$ space.
The perturbation bound gives the remaining half of the error budget.
The stationary routine's workspace is added once for adjacency queries, so
the total space remains
\[
O\!\left((\log n+\log^2\kappa)\log\log n\right).
\]
\end{proof}

\section*{Acknowledgements}\label{sec:ack}

The author would like to thank Richard Peng for extensive discussions and
valuable feedback on the draft.

\paragraph{LLM usage statement.}
 
On September 7, the author prompted GPT-6 Astra with the main theorem statement and initial proof ideas. These arose from the author’s work on another project on dynamic algorithms. The model provided the key algorithmic ideas, including the trimming process and the synchronized analysis.


The author developed the manuscript based on his own interpretation of the
resulting algorithm, with assistance from GPT Codex and Devin.AI in checking
and reformatting.
GPT Codex also assisted with drafting and revising portions of the exposition
and proofs.

The author takes full responsibility for the results and their presentation.

\bibliographystyle{alpha}
\bibliography{ref}

\end{document}